\documentclass[
final
]{dmtcs-episciences}

\usepackage[utf8]{inputenc}
\usepackage{subfigure}

\usepackage{amsthm,amssymb}
\usepackage{amsmath}
\usepackage{latexsym}
\usepackage{comment}
\usepackage{textcomp}
\usepackage{array}
\usepackage{lineno}

\newtheorem{theorem}{Theorem}
\newtheorem{lemma}{Lemma}
\newtheorem{corollary}{Corollary}

\newcommand{\mathcent}{\mbox{\textcent}}
\newcommand{\ket}[1]{\left|#1\right>}
\newcommand{\bra}[1]{\left<#1\right|}
\newcommand{\ep}{$_\varepsilon$}

\author{Masaki Nakanishi\affiliationmark{1}%
	\and Abuzer Yakary{\i}lmaz\affiliationmark{2}
	\and Aida Gainutdinova\affiliationmark{3}
}
\title{New Results on Classical and Quantum Counter Automata\thanks{A preliminary version appeared as ``Masaki Nakanishi, Abuzer Yakary{\i}lmaz:
		Classical and Quantum Counter Automata on Promise Problems. CIAA 2015: 224-237'' \cite{NY15A}. The arXiv number is 1412.6761.}}

\affiliation{Faculty of Education, Art and Science, Yamagata University, Japan\\
	Center for Quantum Computer Science, Faculty of Computing, University of Latvia,   Latvia \\
	Institute of Computational Mathematics and Information Technologies, Kazan Federal University, Russia
}


\received{2016-7-13}
\revised{2018-3-29}
\accepted{2019-9-3}
\begin{document}
\publicationdetails{21}{2019}{4}{13}{1528}
\maketitle
\begin{abstract}
	We show that one-way quantum one-counter automaton with zero-error is more powerful than its probabilistic counterpart on promise problems. Then, we obtain a similar separation result between Las Vegas one-way probabilistic one-counter automaton and one-way deterministic one-counter automaton. 
	
	We also obtain new results on classical counter automata regarding language recognition. It was conjectured that one-way probabilistic one blind-counter automata cannot recognize Kleene closure of equality language [A. Yakary\i lmaz: Superiority of one-way and realtime quantum machines. RAIRO - Theor. Inf. and Applic. 46(4): 615-641 (2012)]. We show that this conjecture is false, and also show several separation results for blind/non-blind counter automata.
\end{abstract}

\section{Introduction}

Quantum computation is a generalization of probabilistic computation which is a generalization of deterministic computation. It is natural to ask whether a quantum model is more powerful than its probabilistic counterpart and similarly whether a probabilistic model is more powerful than its deterministic counterpart. For a fair comparison between these three types of models, bounded-error models of quantum and probabilistic should be considered (as we do in this paper).

Quantum automata models are restricted models of quantum Turing machines, i.e., the type of memory and/or the direction of head movement can be restricted. We find it interesting to determine whether quantum models can have advantage in such restricted case. In this paper, we specifically focus on counter type of memory.
%By analyzing such restricted models, we can find where the power of quantum computation comes from and what kinds of restrictions spoil advantages of quantum computation. In our case, we  investigate how counter resources benefit from quantum computation.

We have a more complete picture for constant-space models (finite state automata) when compared to models using memories (finite state automata augmented with counter(s), stack(s), tape(s), etc.). For example, one-way\footnote{The input is read as a stream from left to right and a single symbol is fed to the machine in each step. We also use two end-markers to allow the machine making some pre- and post-processing.} deterministic finite automata (1DFAs) are equivalent to one-way probabilistic finite automata (1PFAs) and one-way quantum finite automata (1QFAs): they all define the class of regular languages \cite{Sip06,Rab63,LQZLWM12,AG00}. On the other hand, in the case of two-way models\footnote{The input is written on a single-head read-only tape between two end-markers and the head can move in both directions or stay in the same tape square in each step.} abbreviated as 2DFA, 2PFA, and 2QFA, respectively, 2DFAs are equivalent to 1DFAs, 2PFAs are more powerful than 2DFAs, and 2QFAs are more powerful than 2PFAs. \cite{Sh59,Fre81,KW97,AW02}. As a special case, one-way with $\varepsilon$-moves\footnote{It is a restricted version of two-way head such that the head cannot move to the left.} (1\ep) quantum finite automata (1\ep QFAs) can recognize some non-regular languages if the head is allowed to be in a superposition \cite{AI99}. Note that $\varepsilon$-moves can be easily removed for the classical finite automata without increasing the number of states.

When considering finite automata using memory, there are more unanswered cases. The most challenging ones seem to be between quantum and probabilistic models. For example, 1PFAs with a counter (1P1CAs) are more powerful than 1DFAs with a counter (1D1CAs) \cite{Fre79} but we do not know whether 1QFAs with a counter (1Q1CAs) are more powerful than 1P1CAs -- we have only an affirmative answer for one-sided bounded-error \cite{SY12A}. For two-way models, abbreviated respectively 2D1CAs, 2P1CAs, and 2Q1CAs, only 2Q1CAs were shown to be more powerful than 2D1CAs \cite{Yak13B} and the other cases are still open. For one-way pushdown automata models, abbreviated as 1DPDAs, 1PPDAs, and 1QPDAs, respectively, 1DPDAs were shown to be weaker than even Las Vegas restriction of 1PPDAs \cite{HS10} and the question is open between quantum and probabilistic models \cite{YFSA12A}.

%All mentioned results above are regarding language recognition. When considering solving promise problems, a generalization of language recognition such that the aim is to separate two disjoint languages that do not necessarily form the set of all strings, the picture can dramatically change \cite{MNYW05,RY14A,GefY15A,Nak15}. The separation results can be obtained even for one-way models or the case of zero-error -- a very restricted case is that unary QFAs are more powerful than unary PFAs \cite{GaiY15A}. Also as pointed out in \cite{GefY15A}, the effects of randomness and quantumness can be more easily shown with promise problems and some open problems defined on language recognition can be answered in the case of solving promise problems. In \cite{MNYW05,Nak15}, exact\footnote{A single answer is given with probability 1.} 1\ep QPDAs are shown to be more powerful than exact 1\ep PPDAs, which are equivalent to 1\ep DPDAs. In this paper, we obtain the same result between 1Q1CAs and 1P1CAs: we show that exact 1Q1CAs can solve a certain promise problem that cannot be solved by exact 1P1CAs, which are equivalent to 1D1CAs. As mentioned above, Las Vegas 1\ep PPDAs are more powerful than 1\ep DPDAs on language recognition. As the second separation, we obtain the same result between Las Vegas 1P1CAs and 1D1CAs on promise problems. In each separation, we define a new promise problem and give an algorithm for the more powerful model, and then, we present the impossibility result for the weaker model. 

All mentioned results above are regarding language recognition. When considering solving promise problems, the picture can dramatically change. The notion of promise problems was introduced in \cite{ESY84}. Informally, promise problems are a generalization of language recognition such that the aim is to separate two disjoint languages that do not necessarily form the set of all strings. Promise problems have deep connection with fundamental issues in complexity theory such as unique solutions, approximation and complete problems. Readers may refer to \cite{Goldreich06} for a survey on these issues. For quantum computation, the first two notable results on promise problems are the Deutsch-Jozsa algorithm \cite{DJ92,CEMM98} and Simon's algorithm \cite{Simon97}. The results give separation between quantum and classical computation models in the exact and bounded-error settings, respectively. Also for automata models, promise problems have been investigated intensively  \cite{MNYW05,RY14A,GaiY15A,GefY15A,GQZ15,Nak15,BMY17,ZLQG17,GaiY18A}. The separation results can be obtained even for one-way models or the case of zero-error -- a very restricted case is that unary QFAs are more powerful than unary PFAs \cite{GaiY15A}. Also as pointed out in \cite{GefY15A}, the effects of randomness and quantumness can be more easily shown with promise problems and some open problems defined on language recognition can be answered in the case of solving promise problems. In \cite{MNYW05,Nak15}, exact\footnote{A single answer is given with probability 1.} 1\ep QPDAs are shown to be more powerful than exact 1\ep PPDAs, which are equivalent to 1\ep DPDAs. 

In this paper, we show that quantum models can still be more powerful if we replace the stack with a counter: we show that exact 1Q1CAs can solve a certain promise problem that cannot be solved by exact 1P1CAs, which are equivalent to 1D1CAs. As mentioned above, Las Vegas 1\ep PPDAs are more powerful than 1\ep DPDAs on language recognition. As the second separation, we obtain the same result between Las Vegas 1P1CAs and 1D1CAs on promise problems. In each separation, we define a new promise problem and give an algorithm for the more powerful model, and then, we present the impossibility result for the weaker model. As far as the authors know, separation results on neither language recognition nor solving promise problems were known for those automata models. Thus, our separation results on promise problems can be regarded as an important first step toward understanding the complexities of those automata models. 

%The separation results can be obtained even for one-way models or the case of zero-error -- a very restricted case is that unary QFAs are more powerful than unary PFAs \cite{GaiY15A}. Also as pointed out in \cite{GefY15A}, the effects of randomness and quantumness can be more easily shown with promise problems and some open problems defined on language recognition can be answered in the case of solving promise problems. In \cite{MNYW05,Nak15}, exact\footnote{A single answer is given with probability 1.} 1\ep QPDAs are shown to be more powerful than exact 1\ep PPDAs, which are equivalent to 1\ep DPDAs. In this paper, we obtain the same result between 1Q1CAs and 1P1CAs: we show that exact 1Q1CAs can solve a certain promise problem that cannot be solved by exact 1P1CAs, which are equivalent to 1D1CAs. As mentioned above, Las Vegas 1\ep PPDAs are more powerful than 1\ep DPDAs on language recognition. As the second separation, we obtain the same result between Las Vegas 1P1CAs and 1D1CAs on promise problems. In each separation, we define a new promise problem and give an algorithm for the more powerful model, and then, we present the impossibility result for the weaker model. 

Additionally, we present new results on classical counter automata. We disprove the conjecture defined by Yakary{\i}lmaz \cite{Yak12B}: Yakary{\i}lmaz separated 1QFAs with a blind counter from 1DFAs with a blind counter by using the language $ \mathtt{EQ^*} $, the Kleene closure of $ \mathtt{EQ} = \{ a^nb^n \mid n > 0 \} $, and then, the author conjectured  that the same language cannot be recognized by 1PFAs with a blind counter. However, we provide an algorithm for 1PFAs with a blind counter that recognizes $ \mathtt{EQ^*} $. We also show several separation results for blind/non-blind counter automata.

In the next section, we provide the required background and then we present our main results on promise problems in Section 3 and new classical results on language recognition in Section 4. 

A preliminary version of the paper was presented in CIAA 2015 \cite{NY15A}. In this version, we revise the overall paper and added new results and proofs. We modify the definitions of promise problems $ \tt ONE\mbox{-}NONE $ and $ \tt ONE\mbox{-}NONE(t) $ in Section 3.2 since we observe that the argument on the Las Vegas algorithm for $ \tt ONE\mbox{-}NONE(t) $ in \cite{NY15A} is not correct. After this modification, we obtain better success probability in Theorem 3 and we also give correct statement on $ \tt ONE\mbox{-}NONE(t) $ in Theorem 4. Since the promise problem $ \tt ONE\mbox{-}NONE $  is modified, we provide a new impossibility proof for 1D1CAs (Theorem 5). We should remark that this new proof is more complicated (and longer) than the previous proof. Theorem 6 is a new result, which was left open in  \cite{NY15A}. Additionally, we revise the second half of Section 4 and present new results regarding comparisons of classical models: Theorems 8, 9, and 10.

\section{Definitions}

Throughout the paper, $\Sigma$, not containing {\textcent} and \$ (the left and the right
end-markers, respectively), denotes the input alphabet;
$\tilde \Sigma = \Sigma \cup \{\mathcent, \$\}$; $Q$ is the set of (internal) states; $Q_a\subseteq Q$
(resp., $Q_r\subseteq Q$) is the set of accepting (resp., rejecting) states; $q_0$ is the initial state.
For any $w\in \tilde \Sigma^*$, $w(i)$ is the $i$-th symbol of $w$, and $|w|$ is the length of $w$. We assume the reader knows the basics of automata theory. We denote one-way deterministic and nondeterministic finite automata as 1DFA and 1NFA, respectively.

A promise problem $ \tt P = (P_{yes},P_{no}) $ defined on $ \Sigma $ is composed by two disjoint languages $ \mathtt{P_{yes}} \subseteq \Sigma^* $ and $ \mathtt{P_{no}} \subseteq \Sigma^* $, called respectively the set of yes-instances and the set of no-instances. 

A promise problem $ \tt P = (P_{yes},P_{no}) $ is said to be solved by a (probabilistic or quantum) machine $ M $ with error bound $ \epsilon < \frac{1}{2} $ if any yes-instance is accepted with probability at least $ 1 - \epsilon $ and any no-instance is rejected with probability at least $ 1 - \epsilon $. It is also said that $ \mathtt{P} $ is solved by $ M $ with bounded-error. If yes-instances (resp., no-instances) are accepted (resp., rejected) exactly, then it is said that $ \mathtt{P} $ is solved by $ M $ with negative (resp., positive) one-sided error bound $ \epsilon $. If $ \epsilon = 0 $, then it is said that the promise problem is solved exactly. 

A promise problem $ \tt P = (P_{yes},P_{no}) $ is said to be solved by a Las Vegas machine with success probability $ p>0 $ if
\begin{itemize}
	\item any yes-instance is accepted with probability at least $ p $ and is rejected with probability $ 0 $, and,
	\item any no-instance is rejected with probability at least $ p $ and is accepted with  probability $ 0 $.
\end{itemize}
Remark that all non-accepting or non-rejecting probabilities go to the decision of ``don't know''. 

If $ P_{yes} $ is the complement of $ P_{no} $, then conventionally it is said that the language $ P_{yes} $ is recognized by machine $ M $ instead of saying that promise problem $ P $ is solved by  machine $ M $. 

For all models, the input $w\in \Sigma^*$ is placed on a read-only
one-way infinite tape as  $\tilde w=\mathcent w \$$ between the cells indexed by 1 to $|\tilde w|$.
At the beginning, the head is initially placed on the cell indexed by 1 and the value of 
the counter is set to zero. Also, in the following definitions, $m$ denotes the maximum value by which the counter
may be increased or decreased at each step.

A one-way probabilistic one-counter automaton (1P1CA) is a 5-tuple
\[
M=(Q, \Sigma, \delta, q_0, Q_a),
\]
where $\delta: Q\times \Sigma\times \{Z, NZ\} \times Q\times \{-m, ..., m\}\longrightarrow [0,1]$ is a transition function such that $\delta(q, \sigma, z, q', c)=p$ means that
the transition from $q \in Q $ to $q' \in Q$ increasing the counter value by $ c \in  \{-m, ..., m\}  $  occurs with probability $p \in [0,1]$ if the scanned symbol is $\sigma \in \tilde{\Sigma}$ and the status of the counter value is $ z $, where $ Z $ (resp., $  NZ $) means zero (resp., non-zero). The transition function must satisfy the following condition since the overall probabilities must be 1 during the computation: for each triple $ (q \in Q, \sigma \in \tilde{\Sigma}, z \in \{Z,NZ\}), $
\[
\sum\limits_{q'\in Q, c \in \{-m,\ldots, m\}} \delta(q, \sigma, z, q', c) = 1.
\]
The computation is terminated after reading the whole input ($ \mathcent w \$ $) and the automaton accepts (resp., rejects) the input if the final state is in $Q_a$ (resp., $Q\setminus Q_a$). Then, for each input, the accepting (resp., rejecting) probability can be calculated by summing up the probabilities of all the accepting (resp., rejecting) paths.

%A one-way probabilistic blind one-counter automaton (1P1BCA) is a 1P1CA such that it cannot see the status of the counter during the computation and the input is automatically rejected if the value of the counter is non-zero \cite{Gre78}. For 1P1BCAs, the transition function does not depend on the counter value, i.e., $\delta: Q\times \Sigma\times Q\times \{-m, ..., m\}\longrightarrow [0,1]$. The computation is terminated after reading $ \mathcent w \$ $ and the automaton accepts the input if the counter value is zero and the state is in $Q_a$, otherwise it rejects the input. 

A one-way probabilistic blind one-counter automaton (1P1BCA) is a 1P1CA such that it cannot see the status of the counter during the computation and the input is automatically rejected if the value of the counter is non-zero \cite{Gre78}. 
A 1P1BCA is a 5-tuple
\[
M=(Q, \Sigma, \delta, q_0, Q_a),
\]
where $\delta: Q\times \Sigma\times Q\times \{-m, ..., m\}\longrightarrow [0,1]$ is a transition function such that $\delta(q, \sigma, q', c)=p$ means that the transition from $q \in Q$ to $q' \in Q$ increasing the counter value by $c \in  \{-m, ..., m\} $  occurs with probability $p \in [0,1]$
if the scanned symbol is $\sigma \in \tilde{\Sigma}$.
As described above, the transition function must satisfy the following condition: for each pair $( q \in Q, \sigma \in \tilde{\Sigma}), $ 
\[
\sum\limits_{q'\in Q, c\in \{-m, \ldots, m\}} \delta(q, \sigma, q', c) = 1.
\]
The computation is terminated after reading the whole input ($ \mathcent w \$ $) and the automaton accepts the input if the counter value is zero and the state is in $Q_a$, otherwise it rejects the input. The accepting and rejecting probabilities for a given input are calculated as described above.

A configuration of a counter automaton (regardless of whether blind or not) is a pair $(q, v)$ of the current state and the current counter value. Here we do not consider the head position. In our proofs, this will not lead to any confusion.

For each of the above two models, we can define its deterministic version, where the range of
the transition function is restricted to $\{0,1\}$. We abbreviate them respectively as 1D1CA and 1D1BCA.

Moreover, a one-way nondeterministic blind one-counter automaton (1N\-1BCA) can be defined as a 1P1BCA with a special acceptance mode such that it accepts an input if the accepting probability is non-zero and it rejects the input if the accepting probability is zero. Here, each probabilistic choice (the probabilities are insignificant and can be removed) is called as a nondeterministic choice. Then, an input is accepted if and only if there is a path reaching an accepting condition.

Similarly, we can define a one-way universal blind one-counter automaton (1U1BCA), where the automaton accepts the input if the accepting probability is 1 and it rejects the input if the accepting probability is less than 1. In this case, each probabilistic choice (the probabilities are insignificant and can be removed) is called as a universal choice. Then, an input is accepted if and only if every computational path reaches an accepting condition.

A Las Vegas probabilistic machine is a probabilistic machine that (i) never gives a wrong answer but can give a third type of decision called ``don't know'' besides ``accepting'' and ``rejection'' and (ii) both of accepting and rejecting probabilities cannot be non-zero for the same input. For one-way Las Vegas automaton model, we split the set of states into three disjoint sets: the accepting, the rejecting, and neutral states. The automaton says ``don't know'' when it finishes its computation in a neutral state. 

Since quantum computation is a generalization of probabilistic computation \cite{Wat09A}, any quantum model  is expected to simulate its classical counterpart exactly. However, the earlier quantum finite automata (QFAs) models  (e.g. \cite{KW97,MC00}) were defined in a restrictive way and they do not reflect the full power of quantum computation. Even though they were shown to be more powerful than their classical counterparts in some special cases, these QFAs models cannot simulate classical finite automata. The first quantum counter automata model was defined based on these restricted models \cite{Kra99}, and so, they were also shown not to be able to simulate its classical counterparts \cite{YKI05}. Nowadays, we know how to define general quantum automata models that generalize probabilistic automata \cite{Hir10,YS11A}. Therefore, even a superiority result of a restricted model, as given in this paper, serves as a separation between the quantum and probabilistic model. Due to its simplicity, we give the definition of a restricted model that allows to represent our algorithm and we refer the reader to \cite{SY12A} for the definition of general quantum model. We assume the reader familiar with basics of quantum computation. We refer the reader to \cite{NC00} for a complete reference on quantum computation, to \cite{SayY14A} for a short introduction on QFAs, and to \cite{AY15A} for a comprehensive survey on QFAs.

A one-way quantum one-counter automaton (1Q1CA) is a 5-tuple
\[
M=(Q, \Sigma, \delta, q_0, Q_a),
\]
where $\delta: Q\times \Sigma\times \{Z, NZ\} \times Q\times \{-m, ..., m\}\longrightarrow \mathbb{C}$ is a transition function; $\delta(q, \sigma, z, q', c)=p$ means that
the transition from $q$ to $q'$ increasing the counter value by $c$  occurs with probability amplitude $p$
if the scanned symbol is $\sigma$ and the status of the counter value is $ z $.

$\ket{q,v}$ (resp., $\bra{q,v}$), called a ket (resp., bra), denotes the column (resp., row) vector where the entry corresponding to $(q,v)$ is one and the remaining entries are zeros. That is, $\{\ket{q,v}\}$ is an orthonormal basis of $l_2(Q\times\mathbb{Z})$.
For each $\sigma\in \tilde\Sigma$, we define a time evolution operator $U_\sigma$ as follows:
\[
U_\sigma\ket{q, v} = \sum\limits_{(q',c)\in Q\times \{-m, \ldots, m\}} \delta(q,\sigma, z(v), q', c)\ket{q',v+c},
\]
where $z(v)=Z$ (resp., $z(v)=NZ$) if $v=0$ (resp., $v\neq 0$). 
In order to be a well-formed automaton, $U_\sigma$ must be unitary.
The computation of a 1Q1CA is described by 
$\ket{\Psi}=U_{\tilde w(|\tilde w|)}U_{\tilde w(|\tilde w|-1)}\cdots U_{\tilde w(1)}\ket{q_0,0}$.
The following projective measurement $P$ is applied to $\ket{\Psi}$ at the end of the computation:
\[
P=\{P_a=\Sigma_{q\in Q_a, v\in \mathbb{Z}} \ket{q,v}\bra{q,v}, 
P_r=\Sigma_{q\not\in Q_a, v\in \mathbb{Z}} \ket{q,v}\bra{q,v}  \}.
\]
Then, we have ``$a$'' (resp., ``$r$'') with probability $\bra{\Psi}P_a\ket{\Psi}$
(resp., $\bra{\Psi}P_r\ket{\Psi}$).
The automaton accepts (resp., rejects) the input if we have ``$a$'' 
(resp., ``$r$'') as the outcome.

\section{New separation results on promise problems}

We start with the separation of exact quantum model from deterministic one and then we give the separation of  Las Vegas probabilistic model from deterministic one. 
%Lastly, we give our algorithm for 1P1BCA with some new results regarding classical models.

\subsection{Separation of exact 1Q1CAs and 1D1CAs}

We show that there exists a promise problem that can be solved by 1Q1CAs exactly but not by any 1D1CAs. For our purpose, we evaluate XOR value of two comparisons. Let $ a $, $b$, $c$, and $d$ be four even positive numbers. Our first comparison is whether $ a = c $ and the second one is whether $ b=d $, and, our aim is to decide whether 
\[
\left(  \left(  a=c \right) \mathsf{XOR} \left( b=d \right)  	\right)
\]
is true or false. Remark that this expression takes the value of true if and only if exactly one of the comparisons fails.

In order to implement this decision procedure by 1Q1CAs, we give the numbers as $0^a\#0^b\#0^c\#0^d$.
However, due to some technical difficulties, we also append four more numbers as $\#0^{k_1}\#0^{k_2}\#0^{l_1}\#0^{l_2}$, which will help the automaton to set the counter to zero at the end of the computation so that an appropriate quantum interference can be done between the different configurations, i.e., two configurations having different counter values do not interfere. 

\newcommand{\xoreq}{\mathtt{XOR\mbox{-}EQ}}
\newcommand{\xoreqyes}{\mathtt{XOR\mbox{-}EQ_{yes}}}
\newcommand{\xoreqno}{\mathtt{XOR\mbox{-}EQ_{no}}}
Formally, we define our promise problem as follows. Let $\xoreq$ be the set of strings of the form $0^a\#0^b\#0^c\#0^d\#0^{k_1}\#0^{k_2}\#0^{l_1}\#0^{l_2}$
such that
$a, b, c$, and $d$ are even and satisfy the following:
\[
a-c+(-1)^{\delta_{a,c}}(k_1-k_2) = b - d + (-1)^{\delta_{b,d}}(l_1-l_2),
\]
where $\delta_{u,v}=1$ if $u=v$, and $\delta_{u,v}=0$ otherwise.
%\[
%\left\{
%\begin{array}{ll}
%a-c+(k_1-k_2) = b-d+(l_1-l_2)& \mbox{(if $a\neq c$ and $b\neq d$)}\\
%-(k_1-k_2) = b-d + (l_1-l_2)  &\mbox{(if $a=c$ and $b\neq d$)}\\
%a-c+(k_1-k_2) = -(l_1-l_2)  &\mbox{(if $a\neq c$ and $b=d$)}\\
%k_1-k_2 = l_1-l_2  &\mbox{(if $a=c$ and $b=d$).}
%\end{array}
%\right. 
%\]
Then, the set $ \xoreq  $ is our promise. We define yes-instances ($\xoreqyes$) as the set of strings in $ \xoreq $ such that $\left((a=c) \mbox{ xor } (b=d)\right) $ takes the value of true. Then, no-instances ($\xoreqno$) are the ones taking the value of false, or equivalently $ \xoreq \setminus \xoreqyes $.

\begin{theorem}
	\label{theorem:qca_xor_eq}
	The promise problem $\xoreq$  can be solved by 1Q1CAs exactly.
\end{theorem}
\begin{proof}
	We can construct a one-way deterministic reversible one-counter automaton $M_1$,
	which is a special case of the 1Q1CA model,\footnote{A classical reversible operation defined on the set of  configurations is a unitary operator containing only 0s and 1s.} that decides whether $a=c$ as follows.
	\begin{enumerate}
		\item $M_1$ reads the first block $0^a$ and increases the counter
		by one at each transition. 
		\item $M_1$ skips the second block $0^b$.
		\item $M_1$ reads the third block $0^c$ and decreases the counter
		by one at each transition. At the end of this block, $M_1$ decides
		whether $a=c$ or not.
		\item $M_1$ skips the fourth block $0^d$.
		\item $M_1$ reads the fifth block $0^{k_1}$ and increases the counter
		by one if $a\neq c$ (decreases the counter by one if $a=c$) at each
		transition.
		\item $M_1$ reads the sixth block $0^{k_2}$ and decreases the counter
		by one if $a\neq c$ (increases the counter by one if $a=c$) at each
		transition.
		\item $M_1$ skips the seventh and the eighth blocks.
	\end{enumerate}
	Similarly, we can construct a 1Q1CA $M_2$ that decides whether $b=d$
	by comparing $b$ with $d$ using the counter and then the counter is set to zero after reading $0^{l_1}$ and $0^{l_2}$.
	We illustrate $M_1$ and $M_2$ in Figure~\ref{fig:subautomata}.
	
	\begin{figure}
		\begin{center}
			\includegraphics[scale=0.35]{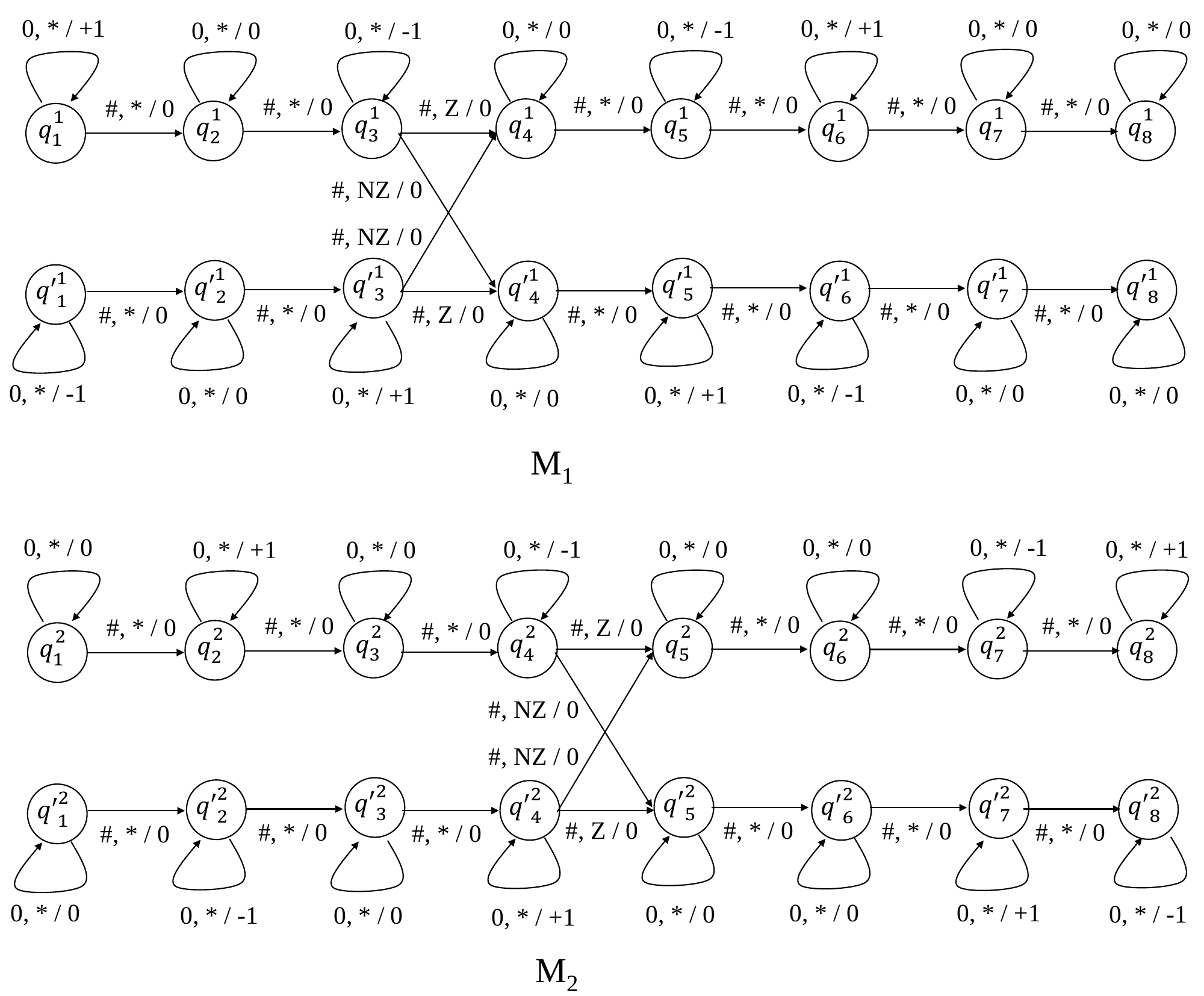}
		\end{center}
		\caption{Subautomata $M_1$ and $M_2$}
		\label{fig:subautomata}
	\end{figure}
	
	In the figure, each label of the edges is of the form $(\sigma, z / c)$, where $\sigma \in \Sigma$, $ z \in\{Z, NZ\}$, and $c \in \{-1, 0, +1\}$.
	A label $(\sigma, z / c)$ means that the transition occurs
	when the input symbol is $\sigma$ and the status of the counter value is $z$ ($*$ denotes a wild card which matches any of $Z$ and $NZ$), and the counter value is updated by $c \in \{-1, 0, +1\}$.
	The initial state is $q^1_1$/$q^2_1$ for $M_1$/$M_2$, respectively.
	The set of accepting states is $\{q^1_8\}$/$\{q^2_8\}$ for $M_1$/$M_2$, respectively. Also the set of rejecting
	states is $\{q'^1_8\}$/$\{q'^2_8\}$ for $M_1$/$M_2$, respectively.
	It is easy to see that if we set the initial state to $q'^1_1$ for $M_1$ ($q'^2_1$ for $M_2$),
	the output is inverted.
	
	We use the algorithm in \cite{CEMM98} (the improved version of Deutsch-Jozsa algorithm\cite{DJ92}) 
	to compute the exclusive-or exactly using the two sub-automata
	as the oracle for Deutsch's problem\cite{Deu85}.
	Note that the counter values are the same between $M_1$ and $M_2$ at the moment
	of reading the last input symbol. Thus, we can construct a 1Q1CA that solves $\xoreq$
	by simulating the improved Deutsch-Jozsa algorithm \cite{CEMM98} on it by running $M_1$ and $M_2$
	in a superposition, which is illustrated in 
	Figure~\ref{fig:DJ_algorithm}. In the figure, the value on each edge represents the amplitude
	associated with the transition. The first and the last transitions occur when it
	reads the left and the right end-markers, respectively. It is straightforward to see that the time evolution operators can be extended to unitary operators by adding dummy states and/or transitions.
\end{proof}

\begin{figure}
	\begin{center}
		\includegraphics[scale=0.30]{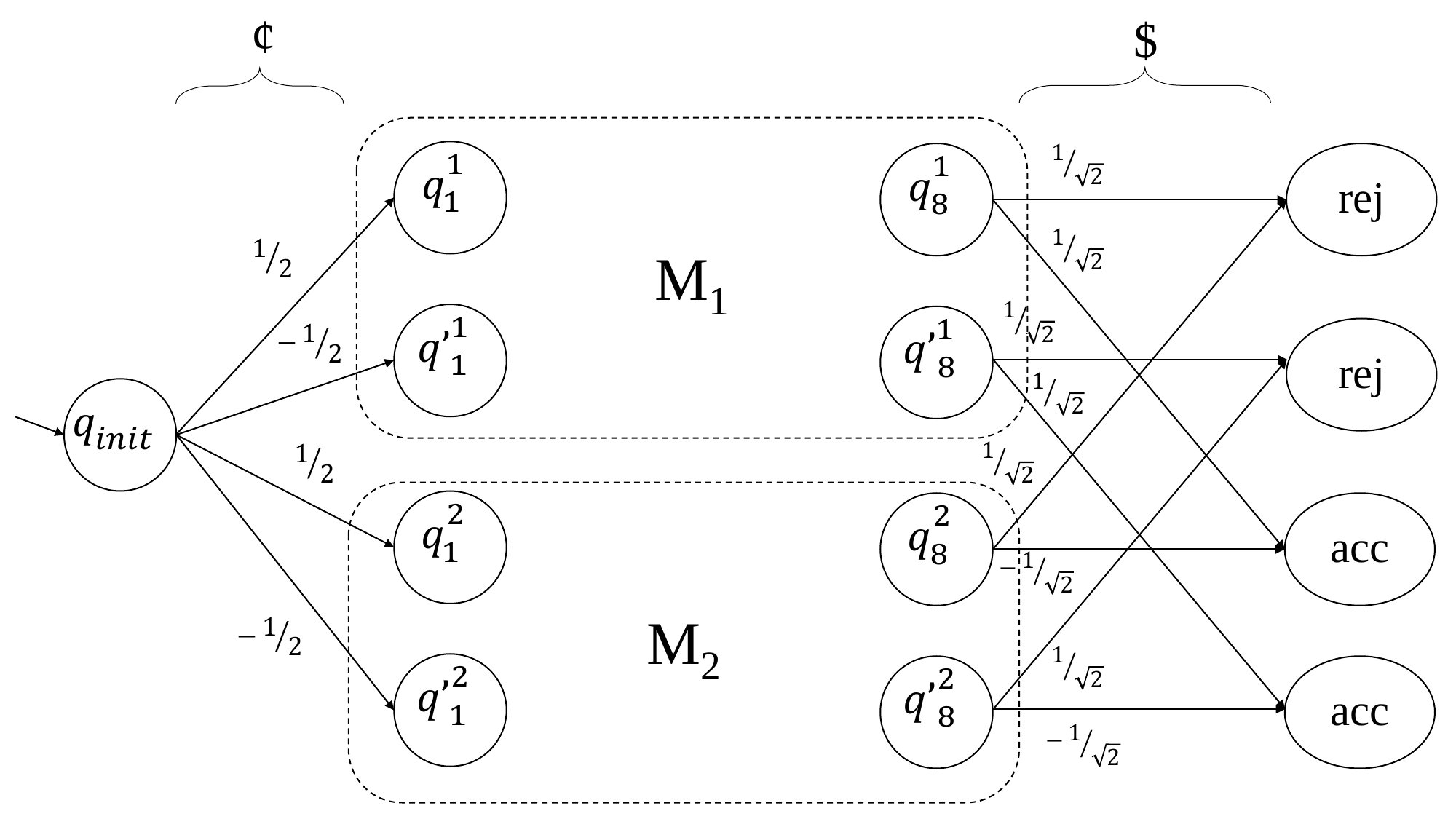}
	\end{center}
	\caption{Simulation of the Deutsch-Jozsa algorithm}
	\label{fig:DJ_algorithm}
\end{figure}

\begin{theorem}
	\label{theorem:impossibility_xor_eq}
	No 1D1CA can solve $ \xoreq $.
\end{theorem}
\begin{proof}
	We assume that there exists a 1D1CA $M$ that solves $\xoreq$. Note that $M$ can have at most $O(n)$ possible configurations for a string whose length is less than $n$, i.e., a constant number of possible states with $O(n)$ possible counter values. Also note that there are $\Theta(n^2)$ possible partial inputs of the form $0^a\#0^b\#$ whose length is less than $n$. Thus, there exist two distinct partial inputs 
	\[
	0^a\#0^b\# \mbox{ and } 0^{a'}\#0^{b'}\#
	\] 
	such that  the configurations after reading them are the same. We will show that there exists a suffix string, 
	\[
	0^c\#0^d\#0^{k_1}\#0^{k_2}\#0^{l_1}\#0^{l_2},
	\]
	such that either
	\[
	u_1 = 0^a\#0^b\#0^c\#0^d\#0^{k_1}\#0^{k_2}\#0^{l_1}\#0^{l_2}
	\] is a yes-instance and
	\[
	u_2=\allowbreak 0^{a'}\#0^{b'}\#0^c\#0^d\#\allowbreak 0^{k_1}\#0^{k_2}\#\allowbreak 0^{l_1}\#0^{l_2}
	\] is a no-instance, or, vice versa. However, $M$ cannot distinguish $ u_1 $ and $u_2$ since the two configurations after reading $0^a\#0^b\#$ and $0^{a'}\#0^{b'}\#$, respectively, are the same.
	This is a contradiction.
	
	Now, we show how to obtain $ u_1 $ and $ u_2 $ as desired. 
	
	We start with the case of $a\neq a'$. We set $l_1$ and $l_2$ to some values providing that
	\[
	d=\frac{b+b'+a-a'}{2}+(l_1-l_2)
	\]
	is even (this is possible since $a$, $b$, $a'$, and $b'$ are even) and $b\neq d$ and $b'\neq d$. Then, we set $k_1$ and $k_2$ to such values providing that
	\[
	-(k_1-k_2)=b-d+(l_1-l_2).
	\]
	Thus, both $ u_1 $ and $ u_2 $, i.e.
	\[
	u_1 = 0^a\#0^b\#0^a\#0^d\#0^{k_1}\#0^{k_2}\#0^{l_1}\#0^{l_2}
	\] 
	and  
	\[
	u_2 = 0^{a'}\#0^{b'}\#0^a\#0^d\#0^{k_1}\#0^{k_2}\#0^{l_1}\#0^{l_2},
	\] become promised input strings since 
	\[
	-(k_1-k_2)=b-d+(l_1-l_2) \mbox{ and } a'-a+(k_1-k_2) = b'-d + (l_1-l_2).
	\]
	In this setting, the former one is a yes-instance and the latter one is a no-instance. 
	
	%The case of $a=a'$ is similar. (See Appendix~C for the details.)
	In the following, we show how to obtain the desired $u_1$ and $u_2$ when $a=a'$. Note that, in this case,  $ b\neq b'$. 
	
	We set $k_1$ and $k_2$ to some values providing that 
	\[
	c=\frac{a+a' + b-b'}{2}+(k_1-k_2)
	\]
	is even (this is possible since $a$, $b$, $a'$, and $b'$ are even) and $a\neq c$ and $a'\neq c$. Then, we set $l_1$ and $l_2$ providing that\
	\[
	a-c+(k_1-k_2) = -(l_1-l_2).
	\] 
	Thus,  both $ u_1 $ and $ u_2 $, i.e., 
	\[
	u_1 = 0^a\#0^b\#0^c\#0^b\#0^{k_1}\#0^{k_2}\#0^{l_1}\#0^{l_2}
	\]
	and
	\[
	u_2 = 0^{a'}\#0^{b'}\#0^c\#0^b\#0^{k_1}\#0^{k_2}\#0^{l_1}\#0^{l_2},
	\]
	become promised input strings since 
	\[
	a-c+(k_1-k_2) = -(l_1-l_2) \mbox{ and } a'-c+(k_1-k_2) = b'-b + (l_1-l_2).
	\] 
	In this setting, again the former one is a yes-instance and the latter one is a no-instance. 
\end{proof}

\subsection{Separation of Las Vegas 1P1CAs and 1D1CAs}

\newcommand{\one}{ \mathtt{ONE} }
\newcommand{\none}{ \mathtt{NONE} }
\newcommand{\onenone}{ \mathtt{ONE\mbox{-}NONE} }
\newcommand{\onenoneyes}{ \mathtt{ONE\mbox{-}NONE_{yes}} }
\newcommand{\onenoneno}{ \mathtt{ONE\mbox{-}NONE_{no}} }

We show that there exists a promise problem that Las Vegas 1P1CAs can solve but 1D1CAs cannot. Our idea is inspired from \cite{RY14A}. 

Let $|w|_a$ be the number of occurrences of the symbol $a$ in the string $w$. We define the sets $ \one $ and $  \none$ as follows. All strings from $\one \cup \none$ have the form $uy$, where $u \in \{a,b,c\}^*,$ $y \in \{d\}^*$, and $|y| \geq |u|$. Moreover, for any string $u y \in \one $, the number of symbols is equal for exactly one pair of $(a,b)$, $(b,c)$, or $(c,a)$, i.e., $|u|_\alpha = |u|_\beta $ for exactly one pair $(\alpha, \beta)\in \{(a,b), (b,c), (c,a)\}$. Also, for any string $u y \in \none $, the number of symbols is equal for none of the pair of $(a,b)$, $(b,c)$, or $(c,a)$, i.e., $|u|_\alpha \neq  |u|_\beta $ for any pair $(\alpha, \beta)\in \{(a,b), (b,c), (c,a)\}$. 

We define a promise problem $ \onenone $, where $ \onenoneyes$ (composed by yes instances) is formed by the concatenation $ \one \cdot \none $ and $\onenoneno $ (composed by no instances) is formed by the concatenation $ \none \cdot \one $. 

\begin{theorem}
	The promise problem $  \onenone $ can be solved by a Las Vegas 1P1CA $P$ with success probability $ p=\frac{1}{3} $.
\end{theorem}
\begin{proof}
	Let $ u_1 y_1 u_2 y_2 $ be a promised input, where $ u_1 \subseteq  \{a,b,c\}^* $, $ y_1 \subseteq \{d\}^* $, $u_2 \subseteq \{a,b,c\}^* $, and $y_2 \subseteq \{d\}^* $. The details of $P$ are as follows. At the beginning, the computation splits into 3 different paths with equal probabilities and each path compares a pair ($(a,b)$, $(b,c)$, or $(c,a)$) in $u_1$. If one of them is succeeded, the input is accepted in that path. All non-accepting paths set their counter to zero by reading $ y_1 $, and then immediately each of them splits into three new different paths with equal probability. Each subpath compares a pair ($(a,b)$, $(b,c)$, or $(c,a)$) in $u_2$. If one of them is succeeded, the input is rejected. Otherwise, $ P $ says ``don't know''.
	
	% The first three paths operate on $ \sigma_1 \gamma_1 $ and the last three paths operate on $ \sigma_2 \gamma_2 $.  Each of the first three paths compares the numbers of symbols for one of the three pairs, $(a,b)$, $(b,c)$, or $(c,a)$, on $\sigma_1$. If it finds equal number of symbols, then the automaton accepts the input on this path. Each of the last three paths makes a similar comparison on $\sigma_2$, but, if it finds equal number of symbols, then the automaton rejects the input on this path. In all the other cases, the automaton says ``don't know''.
	
	If the input is a yes instance, then the numbers of symbols are equal only for a single pair of $ u_1 $. Then the input is accepted with probability $ \frac{1}{3} $ in one of the first three paths, and the computation ends in a neutral state in all the other cases. Similarly, if the input is a no instance, then it is rejected with probability $ \frac{1}{3} $ and the automaton says ``don't know'' with probability $ \frac{2}{3} $.
\end{proof}

\newcommand{\oneornone}{ \mathtt{ONE\textbf{OR}NONE} }

\newcommand{\onenonet}{ \mathtt{ONE\mbox{-}NONE(t)} }
\newcommand{\onenonetyes}{ \mathtt{ONE\mbox{-}NONE_{yes}(t)} }
\newcommand{\onenonetno}{ \mathtt{ONE\mbox{-}NONE_{no}(t)} }

To get a better error bound, we can use the promise problem $ \onenonet $, where yes-instances ($\mathtt{ONE\mbox{-}}$ $\mathtt{NONE_{yes}(t)}$) are formed by $ (\onenoneyes)^t $ and no-ins\-tan\-ces ($\onenonetno$) are formed by $ (\onenoneno)^t $. That is, the error bound can be reduced to $ (\frac{2}{3})^t $  for 1P1CAs, where $ t>1 $.

%%%%%%%%%%%%%%%%%%%%%

\begin{theorem}
	\label{onenone}
	The promise problem $  \onenonet $ can be solved by a Las Vegas 1P1CA $P$ with success probability  $ p=1-(\frac{2}{3})^t $.
\end{theorem}
\begin{proof}
	The details of $P$ are the following. Let $ w=w_1\cdots w_{2t} $ be a promised  input, where for all $i=1, \dots, t$, either $w_{2i-1}\in \one $ and $ w_{2i}\in \none$ or $w_{2i-1}\in \none $ and $ w_{2i}\in \one$. For each part $w_{2i-1}w_{2i}=u_{2i-1}y_{2i-1}u_{2i} y_{2i}$ ($i=1, \dots, t$) of the input, the automaton applies the same strategy as in the previous theorem, where $ u_{2i-1} \subseteq  \{a,b,c\}^* $, $ y_{2i-1} \subseteq \{d\}^* $, $u_{2i} \subseteq \{a,b,c\}^* $, and $y_{2i} \subseteq \{d\}^* $. If the automaton comes to decision ``don't know'', it continues with the next pair until the end.
\end{proof}

%%%%%%%%%%%%%%%%%%%%%

Now, we show that neither $ \onenone $ nor $\onenonet$ can be solved by 1D1CAs. We start with the proof for $ \onenone $, which forms the base for the proof of $ \onenonet $.

\begin{theorem}
	\label{thm_onenone} 
	There is no 1D1CA that can solve promise problem  $\onenone$.  
\end{theorem}
\begin{proof}
	Let us  prove by contradiction and assume that there exists a 1D1CA $M$ that solves $\onenone$. 
	We call two promised inputs  $w $ and $ w'$  non-equivalent with respect to $\onenone$ if either $ w \in \onenoneyes$ and $ w' \in \onenoneno$ or   $ w \in \onenoneno$ and $ w' \in \onenoneyes$. 
	
	We will have a contradiction with our assumption if we show that there exist two strings $w $ and $w' $ that are non-equivalent  with respect to $\onenone$  such that $M$ finishes reading $w$ and $w'$ with the same state and with the same status of counter.
	
	In the proof, we use the following notation. We denote by $c(w)$, $q(w)$, and $v(w)$  the configuration, the state, and the value of the counter of $M$ after reading the partial input  $w$, respectively, and by $\sigma$ an arbitrary symbol of input alphabet. If the automaton reaches the configuration $c'$ from the configuration $c$ when reading $w$, we denote it $c\xrightarrow{w} c'$. Let $m$ be the maximum value by which $M$ may change the value of counter in one step.

	\begin{lemma}
		\label{lemma}
		Let $c(w)=(q, v)$ be the configuration after reading a string $w$. If $M$ starts from $c(w)$ and, for $n\geq|Q|$, all  $v(w), v(w \sigma), v(w\sigma^2), \dots, v(w\sigma^n)$ are non-zero, then  the following is true:
		\begin{enumerate}
			\item there exist $n_1$ and $ n_2$ ($0 \leq n_1<n_2\leq n $ and $ n_1 < |Q|$) such that $q(w\sigma^{n_1})=q(w\sigma^{n_2})$;
			\item there exist numbers $t$ and $r$ ($0 < t\leq|Q| $ and $ 0 \leq |r| \leq m\cdot|Q|$) such that $M$ moves cyclically through some states $q_{i_1}, \ldots, q_{i_t}$ returning to the same state after every $t$ steps, and the value of the counter is changed by the same number $r$ after every $t$ steps as long as the value of the counter is not zero.  
		\end{enumerate}
	\end{lemma}
	\begin{proof}
		Both statements follow from the Pigeon-hole  principle. If the status of counter is the same, then $M$, reading only $ \sigma $'s, is simply a unary automaton and so it always enters a cycle of states after reading more than $ |Q| $ symbols. Thus the first statement is immediate. 
		
		%Because $s(w\sigma^{n_1})=s(w\sigma^{n_2})$ and  behaviour of the automaton depends only on the status and does not depend on the value of the counter then while the status of the counter is non-zero the sequence of states after reading $u^{n_2}$ will be the same as after reading $u^{n_1}$.   
		
		We pick the smallest $ n_1 $ and $n_2$ ($0 \leq n_1 < n_2 \leq n $) such that all $q(w),$ $q(w\sigma),$ $\dots,$ $q(w\sigma^{n_1}),$ $\dots,$ $q(w\sigma^{n_2-1})$ are different and  $q(w\sigma^{n_1})=q(w\sigma^{n_2})$. Then the number $t=n_2-n_1$  is the period of cycle  ($M$ moves cyclically from one state to the next reading $\sigma$ and returns to the same state after every $t$ steps as long as the value of the counter is non-zero). Since after each $t$ steps $M$ is in the same state,  the  counter is changed  by the same value after every $t$ steps as long as the value of the counter is not zero.  
	\end{proof}
	
	Let $c=(q,v)$ be the configuration as given in the lemma. Now we focus on a computation on $M$ reading only $\sigma$'s before the counter hits zero. We  call $t$ and $r$ the period and the difference of the cycle, respectively. Without loss of generality, we assume that $v>0$. 
	%Let's consider a computational process of  $M$ from the configuration $c$ on unary inputs from  $\{u\}^*$ until the moment when the counter's value becomes zero. At these steps the behaviour of the automaton depends only on the state. It follows from Lemma \ref{lemma}  that   if within a certain  finite number of steps $M$ did not reset the counter, then  it necessarily enters into some cycle, where it will stay until the  moment when the value of the counter becomes zero (after this moment the behaviour of the automaton can be changed). Therefore 
	The set $Q$ of states can be divided into disjoint subsets $Q^\sigma_1, \ldots, Q^\sigma_k$ $ ( Q^\sigma_1 \cup \dots \cup Q^\sigma_k=Q )$, where two states $q $ and $ q'$ belong to the same subset $Q^\sigma_j$ iff $M$ moves from $q$ and $q'$ to the same cycle reading $\sigma$'s. We call such partition ${\cal Q}^\sigma=\{Q^\sigma_1, \ldots, Q^\sigma_k\}$ of the set $Q$ as $\sigma$-partition. From Lemma \ref{lemma}, we have that each cycle (and hence each subset from $\sigma$-partition) has two characteristics: its period $t$ and its difference $r$.

	Let $c_1=(q_1,v_1)$ and $c_2=(q_2,v_2)$ be two different configurations. We will say that $c_1$ and $c_2$ are $\sigma$-synchronized if there exists some configuration $c$ and the numbers  $n_1$ and $ n_2 \geq 0$ such that $c_1 \xrightarrow{\sigma^{n_1}} c$ and $c_2 \xrightarrow{\sigma^{n_2}} c$. 
	
	\begin{lemma}
		\label{lemma_synch}
		Let $c=(q,v)$ and $c'=(q,v')$ be two different configurations  with the same state such that $v$ and $ v'$ have the same sign, $|v|,|v'| > m \cdot |Q|$ ($m$ is the maximum value by which the counter can be increased during one step), and $|v-v'|$ is a multiple of $r_\sigma$, where $r_\sigma$  is the difference  of the subset  from $\sigma$-partition ${\cal Q}^\sigma$  that contains $q$. Then $c $ and $c'$ are $\sigma$-synchronized.
	\end{lemma}
	\begin{proof}
		Let $Q^\sigma_j$ be the subset from $\sigma$-partition ${\cal Q}^\sigma$ that contains $q.$ From the definition of $\sigma$-partition, it follows that if $M$ starts its computation in $q$ and reads only $\sigma$'s, then it enters its  cycle in $k$ ($k \leq |Q|$) steps and stays there until the counter hits zero. Let $t_\sigma$ and $r_\sigma$ be the period and the difference of this cycle respectively, $ q' $ be the first visited state in the cycle after the first $k$ steps, and $ r' $ be the value added to the counter in these $k$ steps.  Then we have $c\xrightarrow{\sigma^k}(q', v+r'),$ $c'\xrightarrow{\sigma^k}(q', v'+r'),$ and  $|v-v'|$ is a multiple of $r_\sigma$.
		
		%	Without lose of generality  we assume that $ v > v' $ and then $ v= v'+l \cdot r_\sigma $ for some non-negative integer $l$. Then, we can easily follow 
		%	\[
		%		c  \xrightarrow{\sigma^{k + l \cdot t_u }} (q',v-r'-l \cdot r_\sigma ) = (q',v'-r') \mbox { and }
		%		c'  \xrightarrow{\sigma^{k}} (q',v'-r').
		%	\]	
		%Denote by $c(u^{k_1})=(q', v(u^{k_1}))$ and $c'(u^{k_1})=(q', v'(u^{k_1}))$ the configurations such that $$c \xrightarrow{u^{k_1}} c(u^{k_1}),$$ $$c' \xrightarrow{u^{k_1}} c'(u^{k_1}).$$ It is clear that  $v(u^{k_1}), v'(u^{k_1})\neq 0$. Because the behaviour of $M$ depends on the state, does not depend on the value of the counter and $|v-v'|$ is a multiple of $r_u$, then $|v(u^{k_1})-v'(u^{k_1})|$ also is a multiple of $r_u$.  
		After  reading $\sigma^{k}$, $M$ will return to the  state $q'$  and will increase the  counter by the same value $r_{\sigma}$  after every $t_{\sigma}$ steps as long as the value of the counter is not zero. It is clear that there exists some non-negative integer $l$ such that either $ c \xrightarrow{\sigma^{k+t_{\sigma} \cdot l}}(q', v+r'+l\cdot r_{\sigma})=(q',v'+r')$  or $ c' \xrightarrow{\sigma^{k+t_\sigma \cdot l}}(q', v'+r'+l \cdot r_{\sigma})=(q',v+r')$. 
		Thus, $c$ and $c'$ are $\sigma$-synchronized.  
	\end{proof}
	
	\begin{lemma}\label{lemma1}
		There exists at least one pair of strings ($u_1,u_2)$  such that $u_1\in \one,$ $u_2\in \none$, and 
		{
			\begin{enumerate}
				\renewcommand{\labelenumi}{\rm (\Alph{enumi})}
				\item $c(u_1)=c(u_2) $ or
				\item $c(u_1) = (q(u_1),v(u_1)) \neq c(u_2) = (q(u_2),v(u_2)) $ but $q(u_1)=q(u_2), $ $|v(u_1)|, |v(u_2)| \in \omega(n)$,
				where $n$ is a sufficiently long length.
			\end{enumerate}
		}
	\end{lemma}
	\begin{proof}
		We have two cases. 
		\medskip
		
		{\noindent \bf Case 1:} There exists a symbol  $\sigma \in\{a, b, c\}$ such that $|v( \sigma^n)|\in O(1)$ holds. Without loss of generality, we pick $\sigma = a$. Then, for all inputs $a^n$, we have a constant number of all possible configurations, since the number of states is constant and the possible different values of counter is bounded by $O(1)$. So there exist $n_1$ and $n_2$ $(n_1 < n_2)$ such that $c(a^{n_1}) = c(a^{n_2})$. 
		
		We take $u=b^{n_1}d^{n_1+n_2} $. It is clear that $c(a^{n_1} u)=c(a^{n_2} u)$ therefore $q(a^{n_1} u)=q(a^{n_2} u)$ and $v(a^{n_1} u)=v(a^{n_2} u)$. But  $a^{n_1} u \in \mathtt{ONE}$ and  $a^{n_2} u \in \mathtt{NONE} $.
		\medskip
		
		{\noindent \bf Case 2:} For every symbol $ \sigma \in \{a,b, c\}$, $|v(\sigma^n)|\in \omega(1)$ holds.  We will construct $ u_1 $ and $ u_2 $ in four steps and in each step we define a part of them, i.e.,  
		\[
		u_1 = x_1 x_2 x_3 x_4 \mbox{ and } u_2 = y_1 y_2 y_3 y_4,  
		\]
		where $ x_1,y_1 \in \{a\}^*, x_2,y_2 \in \{b\}^*, x_3,y_3 \in \{c\}^*, \mbox{ and } x_4,y_4 \in \{d\}^* $.
		%step by step. In our construction each string $\sigma \in \{\sigma_1,\sigma_2\}$ will have the form $\alpha\beta\gamma\delta$, where $\alpha \in \{a\}^*, $ $\beta\in \{b\}^*, $ $\gamma\in \{c\}^*, $ and $\delta \in \{d\}^*. $ At each step we will add one of the part ($\alpha, \beta, \gamma,$ or $\delta$) to our  strings $\sigma_1, \sigma_2$.    
		\smallskip
		
		{\noindent \bf Step 1.}
		We pick $r=r(n)$ such that $|v(a^r)|\in \omega(n)$ and we set  $x_1 = y_1 = a^r$. Since $ |v(a^n)| \in \omega(1) $, we can always find such an $ r $ depending on $ n $. Moreover, at each step of a computation, $M$ can increase or decrease the value of the counter by constant amount, and so, for any string $ z $ ($|z|\in O(n)$), we always have $|v(a^rz)|\in \omega(n)$. % after $M$ has already read $a^r $,    the value of the counter at each step of a computation can not be zero during reading $\sigma$ since $|v(a^r)|\in \omega(n)$. 
		\smallskip
		
		{\noindent \bf Step 2.} 
		For $k > |Q|$,  we consider the following sequence of states $q(a^r\,b),$ $q(a^r\,b^2),\dots,$ $q(a^r\,b^k)$. Then, there must exist two distinct non-negative integers $k_1 $ and $ k_2$ ($k_1< k_2 < k$ and $k_1\leq |Q|$) such that all $q(a^r \,b),\dots, q(a^r\,b^{k_1}),\dots, q(a^r\,b^{k_2-1})$ are different and  $q(a^r\,b^{k_2} ) = q(a^r\,b^{k_1})$.    By Lemma \ref{lemma}, the number  $t_b=k_2-k_1$ is the period of cycle and $ r_b $ is the difference of cycle.
		% (after reading $b^{k_1}$, $ M $ returns to the same state for every next $b^{t_b}$) meaning that after $M$ has already read $b^{k_1}$ it will return to the same state after reading every  $b^{t_b}$ until the status of the counter becomes zero. Moreover, $M$ will increase the counter by the same value $r_b$ after each $t_b$ steps reading $b$'s.  
		
		We set $ x_2=b^{k_1}$ and $ y_2=b^{k_2} $. Let $N_a=v(a^r)$, $N_{b_1}=v(a^r \,b^{k_1})-v(a^r),$ and $N_{b_2}=v(a^r \,b^{k_2})-v(a^r)$. Then we have
		\begin{equation*}
		\begin{array}{l}
		q(a^r \, b^{k_1})=q(a^r \, b^{k_2}),\\
		v(a^r \, b^{k_1})=N_a+N_{b_1}, \\
		v(a^r \, b^{k_2})=N_a+N_{b_2}, \mbox{ and}\\
		N_{b_2}-N_{b_1}=r_b.\\
		\end{array}
		\end{equation*}
		\smallskip
		
		{\noindent \bf Step 3.} We set $ x_3 = y_3 = c^{k_1}$. Let  $N_{c}=v(a^r\, b^{k_1}\,c^{k_1})-v(a^r\, b^{k_1}) $. Then, we have the followings: 
		%Since the counter does not hit to zero, we can follow that:
		% behaviour of $M$ depends only on the status, does not depend on the counter's value and because $s(a^r\, b^{k_1})=s(a^r\, b^{k_2})$ we also have $v(a^r\, b^{k_2}\,c^{k_1})-v(a^r\, b^{k_2})=N_c.$ Then 
		\begin{equation}\label{step3}
		\begin{array}{l}
		q(a^r \, b^{k_1}\, c^{k_1})=q(a^r \, b^{k_2}\, c^{k_1}),\\
		v(a^r \, b^{k_1}\, c^{k_1})=N_a+N_{b_1}+N_c, \\
		v(a^r \, b^{k_2}\, c^{k_1})=N_a+N_{b_2}+N_c, \mbox{ and}\\
		v(a^r \, b^{k_2}\, c^{k_1})-v(a^r \, b^{k_1}\, c^{k_1})=r_b.\\
		\end{array}
		\end{equation}
		
		{\noindent \bf Step 4.} 
		Let  $Q^d_j$ be the subset from $d$-partition ${\cal Q}^d$ that contains $q(x_1 x_2 x_3) = q(y_1 y_2 y_3)$ and then $ t_d$ and $r_d$ be the period and the difference of the cycle, respectively. 
		%We denote   $\sigma_{1}=a^r\,b^{k_1}\,c^{k_1}$, $\sigma_{2}=a^r\,b^{k_2}\,c^{k_1}$. Consider a subset $Q^d_j$ from $d$-partition ${\cal Q}^d$ that contains $s(\sigma_1)$.  Let $t_d$ and $r_d$ be a period and a difference of this subset respectively.
		Depending on the values $r_b$ (from the step 2) and $r_d$, we can have different cases. We denote $ x_1 x_2 x_3 $ by $   x_{123} $ and  $ y_1 y_2 y_3 $ by $ y_{123} $.
		
		The case of $r_b=0$: $v(x_{123})=v(y_{123})$ and so $c(x_{123})=c(y_{123})$.
		% then from relations (\ref{step3}) we get that after step 3  $v(\sigma_1)=v(\sigma_2)$ and therefore $c(\sigma_1)=c(\sigma_2)$. 
		We set $x_4=y_4=d^l$ such that $|x_4| \geq | y_{123} | > | x_{123} | $.
		Then $c(x_{123}x_4)=c(y_{123}y_4)$ but $ x_1 x_2 x_3 x_4 \in \one$ and $ y_1 y_2 y_3 y_4 \in \none$.
		
		The case of $r_b \neq 0$ and $r_d \neq 0$: If $|v( x_{123})-v( y_{123} )| = r_b $ is a  multiple of $r_d$, then due to Lemma \ref{lemma_synch}, we can conclude that  $c(x_{123})$ and $c(y_{123})$ are $d$-synchronized. We set $ x_4=d^{l_1+l}$ and $y_4 = d^{l_2+l}$, where $l_1$ and $l_2$ are the numbers of steps that are needed to synchronize the configurations $c(x_{123})$ and $c(y_{123})$ respectively, and $l$ is the value providing that $| x_4 |\geq | x_{123} |$ and $| y_4 |\geq | y_{123} |$. Thus, we can follow that $ c(x_{123}x_4) = c(y_{123}y_4) $ but  $ x_1 x_2 x_3 x_4 \in \one$ and $ y_1 y_2 y_3 y_4 \in \none$.
		%\begin{equation*}
		%\begin{array}{l}
		%c(\sigma_1\delta_1)=c(\sigma_2\delta_2),\\
		%\sigma_1\delta_1\in \one,\\ 
		%\sigma_2\delta_2 \in \none.
		%\end{array}
		%\end{equation*}
		
		If $ r_b $ is not multiple of $ r_d $, then we can re-define $ y_2 $ as $ b^{k_1+t_b \cdot r_d} $ by setting $ k_2 = k_1 + t_b \cdot r_d $. Then, Equations \ref{step3} can be rewritten as
		% $|v(\sigma_1)-v(\sigma_2)|$ is not  a  multiple of $r_d$.  Then we can correct the part $\beta_2$ to achieve this property as follows: $\beta_2=b^{k_1+t_b \cdot r_d}.$ Then  from (\ref{step3}), by using $k_2=k_1+t_b\cdot r_d$, we get 
		\begin{equation*}
		\label{step4}
		\begin{array}{l}
		v(a^r \, b^{k_1}\, c^{k_1})=N_a+N_{b_1}+N_c, \\
		v(a^r \, b^{k_2}\, c^{k_1})=N_a+N_{b_1}+r_b \cdot r_d+N_c,\\
		v(a^r \, b^{k_2}\, c^{k_1})-v(a^r \, b^{k_1}\, c^{k_1})=r_b\cdot r_d.\\
		\end{array}
		\end{equation*}
		This update on $ y_2 $ concludes that $c(x_{123})$ and $c(y_{123})$ are $d$-synchronized as described above, and so $ c(x_{123}x_4) = c(y_{123}y_4) $ but  $ x_1 x_2 x_3 x_4 \in \one$ and $ y_1 y_2 y_3 y_4 \in \none$.
		
		The case of $r_b\neq 0$ and $r_d=0$: 
		% In this instance after the automaton enters the cycle, it does not change the counter after  reading every $d^{t_d}$. 
		We set $ x_4 = y_4 = d^{l_1+l}$ such that $l_1$ is the minimum numbers of steps that is sufficient to enter the cycle and $l$ is the minimum value providing that $| y_4 | \geq | y_{123} | > |x_{123} | $. Thus, we can follow that $ q(x_{123}x_4) = q(y_{123}y_4) $ and $ | v(x_{123}x_4) | \in \omega(n) $ and $ | v(y_{123}y_4) | \in \omega(n) $, but $ x_1 x_2 x_3 x_4 \in \one$ and $ y_1 y_2 y_3 y_4 \in \none$.
		%\begin{equation*}
		%\begin{array}{l}
		%s(\sigma_1\delta)=s(\sigma_2\delta),\\
		%|v(\sigma_1\delta)|\in \omega(n), |v(\sigma_2\delta)|\in \omega(n),\\
		%\sigma_1\delta\in \one,\\ 
		%\sigma_2\delta\in \none.
		%\end{array}
		%\end{equation*}
	\end{proof} 
	
	Now, we construct the pair of strings $w $ and $ w'$ that are non-equivalent with respect to $\onenone$. Due to Lemma \ref{lemma1}, there exist two strings $ u_1$ and $ u_2$ such that $ u_1 \in \one$ and $ u_2 \in \none$, and 
	\begin{enumerate}
		\renewcommand{\labelenumi}{\rm (\Alph{enumi})}
		\item $c(u_1)=c(u_2) $ or
		\item $c(u_1)\neq c(u_2) $ but $q(u_1)=q(u_2) $ and $|v(u_1)|, |v(u_2)| \in \omega(n)$.
	\end{enumerate}
	For simplicity, we call the pair $u_1$, $ u_2 $ as A-type if it  satisfies Condition A and we call it as B-type if the Condition B is satisfied.
	
	If the pair $u_1$ and $u_2$ is A-type, then, by assuming $ c(u_1) $ is the initial configuration, we can construct two new strings $ u_1' $ and $ u_2' $ as described above such that $ u_1' \in \none $ and $ u_2' \in \one $, and then, the pair $ w=u_1 u_1' $ and $ w'=u_2 u_2' $ is either A-type or B-type. Thus, $ M $ gives the same decisions for $ w $ and $ w'$ but $ w \in \onenoneyes $ and $ w' \in \onenoneno $.
	
	If the pair $ u_1 $ and $u_2$ is B-type, then we can define $ u_1' $ and $ u_2' $ as follows. Since the value of counter is superlinear in $n$, there exist two minimal non-negative integers $ k_1 $ and $ k_2 $ such that $ k_1 < k_2 $, $ k_1 \leq |Q| $ and $q(u_1 a^{k_1}) = q(u_2 a^{k_2})$. We set $ u_1' = a^{k_2}b^{k_1} d^{k_1+k_2} $ and $ u_2' = a^{k_1}b^{k_1} d^{k_1+k_2} $. It is clear that (1) $ u_1' \in \none $ and $ u_2' \in \one $, and (2) $ q(u_1 u_1') = q(u_2 u_2') $ and the values of counter are superlinear in $n$ ($\in w(n)$) for $ w= u_1 u_1' $ and $ w' = u_2 u_2' $.  Thus, $ M $ gives the same decisions for $ w $ and $ w'$ but $ w \in \onenoneyes $ and $ w' \in \onenoneno $.
	
	With the contradiction in each possible case, we can conclude that there is no 1D1CA solving $ \onenone $.
\end{proof}

We can easily extend the above impossibility proof for the promise problem $\onenonet$ ($ t > 1 $).

\begin{theorem} 
	There is no 1D1CA that can solve promise problem  $\onenonet$  for any  $t \in \mathbb{Z}^+$.
\end{theorem}
\begin{proof}
	Let $ M $ be a 1D1CA  solving $\onenonet$. In the previous proof, when starting in some configuration, say $ c(u) $ for some strings $u$, we show how two construct two different strings, say $u_1$ and $u_2$, such that $ u_1 \in \one $, $ u_2 \in \none $, and $ M $, after reading $ u_1 $ and $ u_2 $, ends with either the same configuration ($ c(u u_1) = c(u u_2) $) or with the same state ($ q(u u_1) = q(u u_2) $) and with some values of counter that are superlinear for some sufficiently big $ n $ (i.e., $ |v(u u_1)|, |v(u u_2)| \in \omega(n) $). We call the former pair as A-type and the latter pair as B-type. For B-type pairs, we assume here that the values of counter are not allowed to be less than a value quadratic in $ n $ and $ t $, i.e., $ |v(u u_1)|, |v(u u_2)| \in \Omega(t^2 n^2) $.
	
	Based on these facts, we can construct the following two strings
	\[
	w = w_1 w_2 \cdots w_t \mbox{ and }
	w' = w_1' w_2' \cdots w_t'
	\]
	such that each $ w_j \in \onenoneyes $, each $ w_j' \in \onenoneno $ ($ 1 \leq j \leq t $), but $ M $, after reading $ w $ and $ w' $, finishes its computation in the same state and in the same status of counter. Thus, $ M $ gives the same decision for a yes-instance and for a no-instance, which leads us to conclude that $ M $ cannot solve $ \onenonet $.
	
	We start from the initial configuration. Each $ w_j $ ($ w_j' $) is composed by two strings $ u_j y_j $ ($ u_j' y_j' $) such that  $ u_j, y_j' \in \one $ and $ y_j,u_j' \in \none $. For $ j=1,\ldots,t $, we first construct $ u_j $ and $ u_j' $, and then $ y_j $ and $ y_j' $. If all pairs are A-type and then the construction is straightforward since M ends in the same configurations after reading A-type pairs and from the same configuration we can always construct two new pairs as desired. 
	
	If we obtain a B-type pairs at some point of the construction, we can define the remaining parts of $w$ and $w'$ as we do at the end of the previous proof. First, we can be sure that the values of counter are quadratic in $t$ and $n$ ($ |v(u u_1)|, |v(u u_2)| \in \Omega(t^2 n^2) $). Then, the length each new obtained pair can be easily bounded by $ O(n) $. That means the status of the counter will be the same for the remaining of the computation and so $ M $ behaves like a deterministic finite automaton. Thus, it is very easy to fool $ M $ when constructing the remaining pairs that requires equality checks. 
\end{proof}

\section{New results on classical counter automata}

In this section, we show the results that separate the expressive power of several models of blind/non-blind counter automata. For this purpose, we denote the class of languages recognized by a {\it Model} as $\mathcal{L}(\mbox{{\it Model}})$.

First of all, we present a 1P1BCA algorithm for the Kleene closure of equality language:
\[
\mathtt{EQ^*}=\{ \varepsilon \} \cup \{ a^{n_1}b^{n_1}\cdots a^{n_k}b^{n_k} | n_i>0 (1\leq i \leq k), k\geq 1  \},
\]
which was shown not to be recognized by any one-way deterministic finite automaton with multi blind counters \cite{Gre78}. 
Recently, Yakary{\i}lmaz presented a negative one-sided error 1Q1BCA algorithm for this language and he conjectured that it cannot be recognized by 1P1BCAs \cite{Yak12B}. Now, we show that this conjecture is false. It is also surprising that our new algorithm is kind of a probabilistic adaptation of the quantum algorithm given by Yakary{\i}lmaz.

\begin{theorem}
	\label{theorem:rtP1BCA_EQ}
	The language $ \mathtt{EQ^*} $ can be recognized by a 1P1BCA $ M $ with negative one-sided error bound $ \frac{1}{3} $.
\end{theorem}
\begin{proof}
	We assume that the input is of the form $a^{n_1}b^{m_1}\cdots a^{n_k}b^{m_k}$. Otherwise, $M$ rejects the input deterministically (exactly). At the beginning of each block $a^{n_l}b^{m_l}$ ($1\leq l \leq k$), $M$ selects one of the following three paths ($Path_i$'s) with equal probability:
	\[
	\begin{array}{ll}
	Path_i (1\leq i \leq 3):& \mbox{ $M$ increases (resp., decreases) the counter by $i$}\\
	& \mbox{each time 
		reading an $a$ (resp., a $b$) of the block.}
	\end{array}
	\]
	The computation always ends in an accepting state (except the deterministic check mentioned at the beginning). Thus, the input is accepted if and only if the value of counter is zero. It is obvious that $M$ accepts any member of $ \tt EQ^*$ with certainty. We consider the case that the input $w \not\in \tt EQ^*$. Let $i_{max}$ be the greatest index satisfying $n_{i_{max}}\neq m_{i_{max}}$, i.e., $a^{n_{i_{max}}}b^{m_{i_{max}}}$ is the last block satisfying $n_{i_{max}}\neq m_{i_{max}}$. Let $ path' $ be a probabilistic path before reading the $ i_{max} $-th block having the counter value $c$. This path will split into three sub-paths $ subpath'_1 $, $ subpath'_2 $, and $ subpath'_3 $ and each subpath reads the block as described above. Let $ c_1 $, $c_2$, and $c_3$ be the counter values of these sub-paths, respectively, after reading the block.  Any computation starts from $ subpath'_i $ will have the same counter value of $ c_i $ at the end of the computation since the remaining blocks have the same numbers of $a$'s and $b$'s, where $ 1 \leq i \leq 3 $. Assume that $ subpath'_i $ leads to a decision of acceptance. This is possible only if $ c_i =0$. Let $ d = n_{i_{max}} - m_{i_{max}} \neq 0 $. Then the values of $ c_1 $, $ c_2 $, and $c_3$ are $ c+d $, $ c+2d $, and $c+3d$, respectively. Therefore, only one of them can be zero. That is, the maximum accepting probability that $ path' $ can contribute is $ \frac{1}{3} $. This is the case also for all other probabilistic paths that exist just before reading the $ i_{max} $-th block. Therefore the overall accepting path can be bounded by $ \frac{1}{3} $.
\end{proof}

It is clear from the analysis given in the proof that the error bound can be reduced to $ \frac{1}{k} $ for any $ k $ by spiting into $ k $ probabilistic paths on each block instead of $3$.
\begin{corollary}
	The language $ \mathtt{EQ^*} $ can be recognized by a 1P1BCA $ M $ with any negative one-sided error bound $ \epsilon \leq \frac{1}{2} $.
\end{corollary}

Remark that since any language recognized be a 1P1BCA with negative one-sided error is recognized by 1U1BCA, we can also conclude that $ \tt EQ^* $ can be recognized by 1U1BCAs.

Even though any number of blind counters is useless for a 1DFA (or a 1NFA) \cite{Gre78}, a single non-blind counter is enough in order to recognize $ \tt EQ^* $, i.e., 1D1CAs can recognize $ \tt EQ^* $. These facts together with Theorem~\ref{theorem:rtP1BCA_EQ} imply that $\mathcal{L}(1D1BCA) \subsetneq \mathcal{L}(1P1BCA) \cap \mathcal{L}(1D1CA)$. 
Another related result is that Freivalds \cite{Fre79} proved that $ \mathtt{EQ^3} = \{ c^nd^ne^n \mid n \geq 0 \} $ can be recognized by 1P1BCAs with arbitrary small negative one-sided error bound and this non-context free language, of course, cannot be recognized by a 1D1CA. 
We represent our result with the known facts in Figure~\ref{fig:summary2}. We still do not know whether bounded-error 1Q1BCAs are more powerful than bounded-error 1P1BCAs.

\begin{figure}
	\begin{center}
		\includegraphics[scale=0.28]{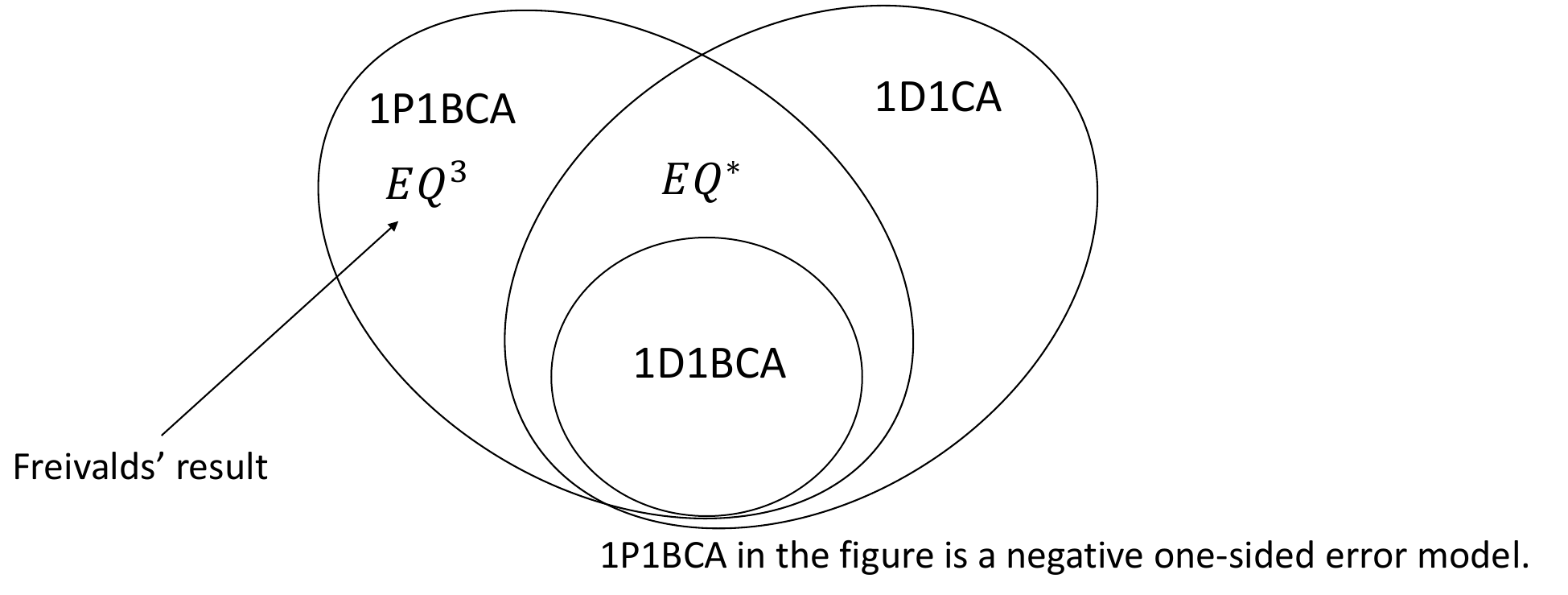}
	\end{center}
	\caption{Separation between $\mathcal{L}(1D1BCA)$ and  $\mathcal{L}(1P1BCA) \cap \mathcal{L}(1D1CA)$}
	\label{fig:summary2}
\end{figure}

Our next result is on incomparability of 1U1BCAs and 1D1CAs. 
In order to show it, we consider the complement of $\mathtt{EQ^*}$:
\[
\overline{\mathtt{EQ^*}}=\overline{\{a^{n_1}b^{n_1} \ldots a^{n_m} b^{n_m} | m\geq 0\}}.
\]

\begin{theorem}
	1U1BCAs cannot recognize $\overline{\mathtt{EQ^*}}$.
\end{theorem}
\begin{proof}
	We assume that there exists a 1U1BCA $M$ that recognizes $\overline{\mathtt{EQ^*}}$. Let $Q$ be the set of states of $M$. Let $w=a^{n_1}b^{n_1}\cdots a^{n_m}b^{n_m}$ be a string not in $\overline{\mathtt{EQ^*}}$ where $n_1 >|Q|$ ($ w \in \tt EQ^* $). Then, there exists a rejecting path for $w$, say $p$. We consider the computation along the path $p$. 
	
	Since $n_1>|Q|$, there exists a state $s\in Q$ such that $M$ enters $s$ at least twice when reading $a^{n_1}$. We assume that $M$ enters the state $s$ just after reading $a^{t}$ and $a^{t'}$ ($0<t<t'<n_1$) in the first block. In other words, $M$ enters the state $s$ after the $t$-th step and the $t'$-th step. We divide the path $p$ into three subpaths $ p = p_1 \cdot p_2 \cdot p_3$ where $p_2$ starts from $(t+1)$-th step and finishes at the $t'$-th step. Then both of $ p' = p_1\cdot p_2 \cdot p_2 \cdot p_3$  and $ p'' = p_1\cdot p_2 \cdot p_2 \cdot p_2 \cdot p_3$ are valid computation paths for input strings $w_1=a^{n_1+(t'-t)}b^{n_1}\ldots a^{n_m}b^{n_m}$ and $w_2=a^{n_1+2(t'-t)}b^{n_1}\ldots a^{n_m}b^{n_m}$, respectively. Note that $w_1, w_2\in \overline{\mathtt{EQ^*}}$. Also, note that both of $ p' $ and $ p'' $ lead to the same final state as $p$. Then, at least one of $ \{p',p''\} $ has non-zero counter value or both of them has the same counter value as $p$ at the end of the computation. This is because if the counter value increases by $d (\neq 0)$ along with $p_2$, then the final counter values are different for $p'  $ and $ p'' $. If $d=0$, then the counter values are the same for $p$, $p'$, and $p''$.  Therefore, at least one of $ \{p',p''\} $ is a rejecting path. (Remember that, by the definition of blind counter automata, computation that ends with a non-zero counter value is always rejected.) However, both of $w_1$ and $w_2$ are in $\overline{\mathtt{EQ^*}}$. This is a contradiction.
\end{proof}

%{\color{red}{Is this proof extended to case $k$-blind counters?}}

%{\color{blue}{Can a similar proof be followed for 1N$k$BCAs? One problem when considering shorter and longer strings some new nondeterministic paths can appear! Unary 1NFAs can have a normal form (e Chrobak normal form), at the beginning of the computation, it makes a single nondeterministic choice and then each path has an (possibly empty) initial part following by a cycle. If we have a similar structure for1N$k$BCAs, then we can follow a proof more easily!}}

Figure~\ref{fig:summary3} summarizes the incomparability of 1U1BCAs and 1D1CAs. Note that 1D1CA can recognize $\overline{\mathtt{EQ^*}}$ and a negative one-sided bounded-error 1P1BCA algorithm is also a 1U1BCA algorithm since all members accepted with probability 1.

\begin{figure}
	\begin{center}
		\includegraphics[scale=0.28]{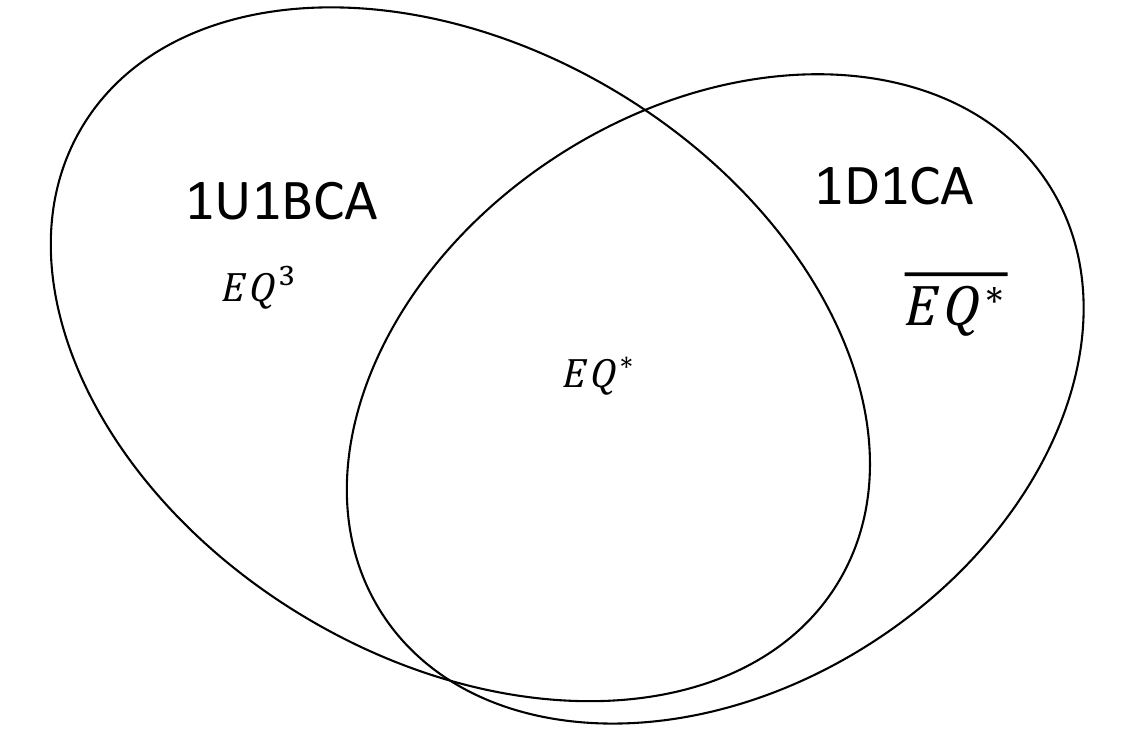}
	\end{center}
	\caption{Incomparability of 1U1BCAs and 1D1CAs}
	\label{fig:summary3}
\end{figure}

Our last result is the separation between $\mathcal{L}(1P1CA)$ and $\mathcal{L}(1U1BCA)\cup \mathcal{L}(1D1CA)$. In order to show it, we define language $L$ as follows:
\[
L=\overline{\mathtt{EQ^*}} \cup \mathtt{EQ^3} \mbox{ where }\mathtt{EQ^3}=\{c^n d^n e^n | n\geq 0\}.
\]
Then, we have the following theorems.
\begin{theorem}
	\label{theorem:1P1CA_L}
	$L$ can be recognized by a 1P1CA with negative one-sided bounded error.
\end{theorem}

\begin{proof}
	Let $M_{\overline{\mathtt{EQ^*}}}$ be a 1D1CA that recognizes $\overline{\mathtt{EQ^*}}$. Also, let $M_{\mathtt{EQ^3}}$ be a 1P1CA that recognizes $\mathtt{EQ^3}$ with negative one-sided bounded error. In order to recognize $L$, we use $M_{\overline{\mathtt{EQ^*}}}$ and $M_{\mathtt{EQ^3}}$ as subautomata. Note that $\overline{\mathtt{EQ^*}} \subset \{a,b\}^*$ and $\mathtt{EQ^3} \subset \{c,d,e\}^*$. Thus, the following 1P1CA $M$ recognizes $L$:
	\begin{itemize}
		\item if the input is an empty string, $M$ accepts it.
		\item if the first symbol is $a$ or $b$, $M$ executes $M_{\overline{\mathtt{EQ^*}}}$.
		\item otherwise, $M$ executes $M_{\mathtt{EQ^3}}$.
	\end{itemize}
\end{proof}

\begin{theorem}
	\label{theorem:1D1CA_1U1BCA_L}
	Neither 1D1CAs nor 1U1BCAs can recognize $L$.
\end{theorem}
\begin{proof}
	If there exists a 1D1CA that recognizes $L$, then it can be regarded as a 1D1CA that recognizes $\mathtt{EQ^3}$ by ignoring transitions for the symbols $a$ and $b$. This is a contradiction.
	
	Similarly, if there exists a 1U1BCA that recognizes $L$, then it can be regarded as a 1U1BCA that recognizes $\overline{\mathtt{EQ^*}}$ by ignoring transitions for the symbols $c$, $d$, and $e$. This is a contradiction.
\end{proof}

By Theorems~\ref{theorem:1P1CA_L} and \ref{theorem:1D1CA_1U1BCA_L}, we can separate $\mathcal{L}(1P1CA)$ and $\mathcal{L}(1U1BCA)\cup \mathcal{L}(1D1CA)$ as illustrated in Figure~\ref{fig:summary4}.

\begin{figure}
	\begin{center}
		\includegraphics[scale=0.28]{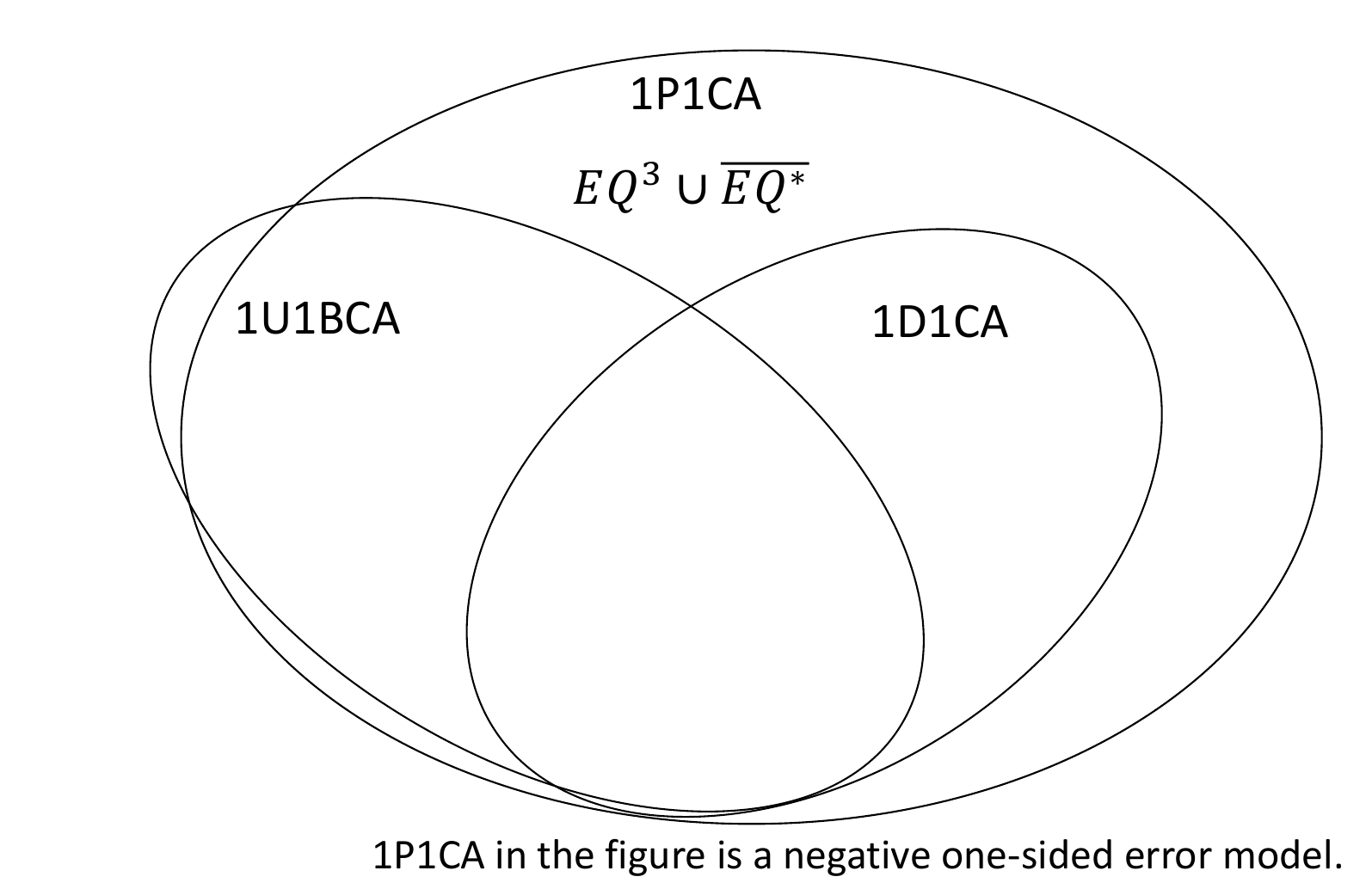}
	\end{center}
	\caption{Separation between $\mathcal{L}(1P1CA)$ and $\mathcal{L}(1U1BCA)\cup \mathcal{L}(1D1CA)$}
	\label{fig:summary4}
\end{figure}

As pointed before languages recognized by 1P1BCAs are also recognizable by 1U1BCAs. Thus, combining all the above results, we have the hierarchy of the models as illustrated in Figure~\ref{fig:summary}.

\begin{figure}
	\begin{center}
		\includegraphics[scale=0.28]{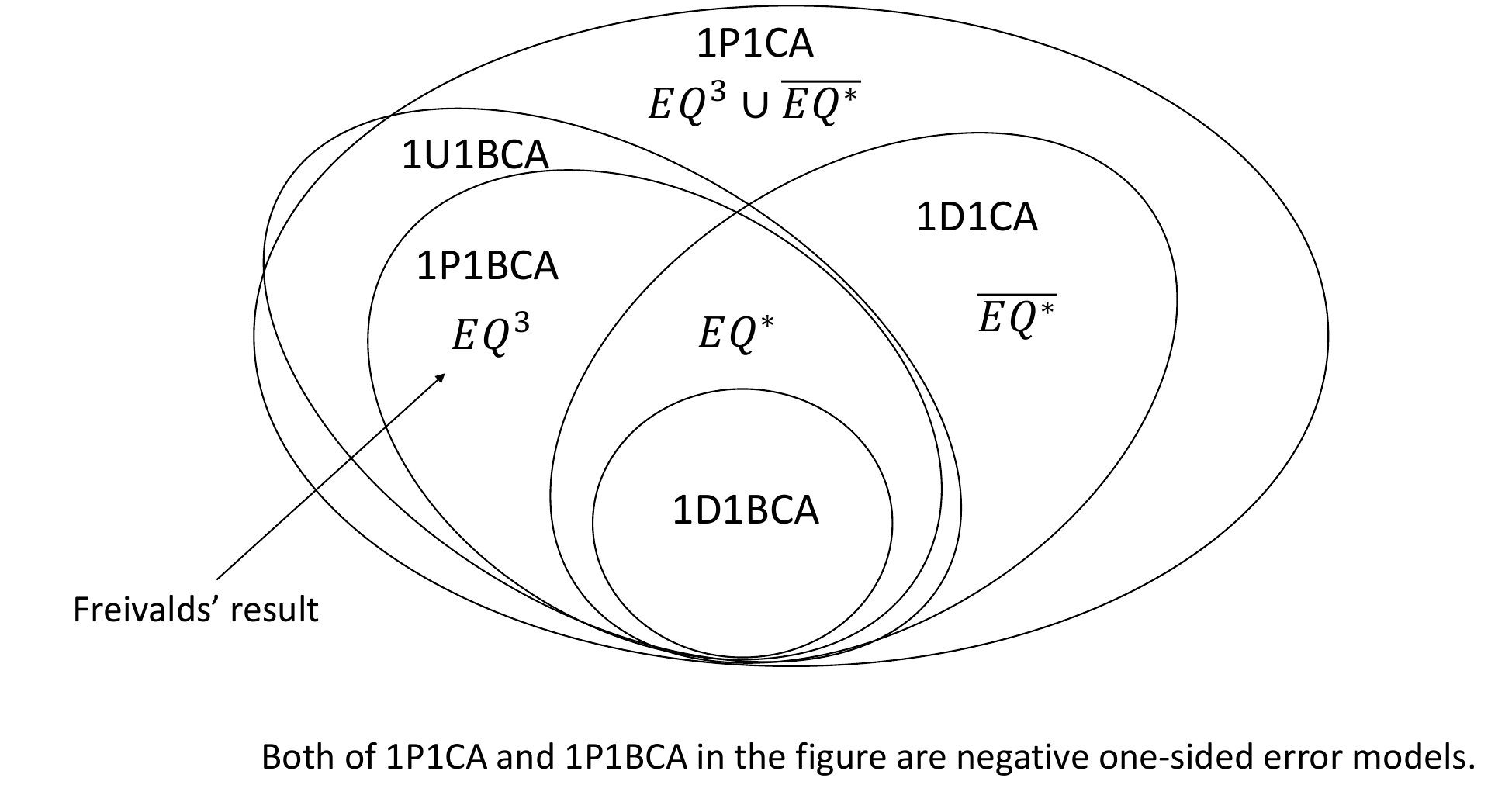}
	\end{center}
	\caption{Hierarchy of various models of counter automata}
	\label{fig:summary}
\end{figure}

% We close the section with some discussions on nondeterministic and universal models. The language $ \mathtt{EQ^*} $ can be recognized by a 1U1BCA: it universally picks each block of $ a^+b^+ $ and then deterministically determines whether the numbers of $a$'s and $b$'s are equal. The input is accepted if and only if each block has equal number of $ a $'s and $ b $'s. One may ask whether we can recognize the complement of $ \mathtt{EQ^*} $. Here any input that is not of the form $ (a^+b^+)^+ $ can be deterministically detected. 

As a future work, we find interesting to identify whether there is an alternation hierarchy for one-way blind-counter automata with and without $\varepsilon$-moves.

\acknowledgements
We thank Klaus Reinhardt for answering our question regarding the subject matter of this paper and anonymous reviewers for their helpful comments. 

Some parts of this work was done while Gainutdinova was visiting National Laboratory for Scientific Computing, Petr\'{o}polis, RJ, 25651-075, Brazil in June 2015, supported by CAPES with grant 88881.030338/2013-01.

Masaki was partially supported by JSPS KAKENHI Grant Numbers 16K00007, 24500003 and 24106009, and also by the Asahi Glass Foundation. Yakary{\i}lmaz was partially supported by CAPES with grant 88881.030338/2013-01 and ERC Advanced Grant MQC.

%\nocite{*}
%\bibliographystyle{abbrvnat}
% use the following instead if you encounter problems 
\bibliographystyle{alpha}
\bibliography{tcs}
\label{sec:biblio}

\end{document}